\PassOptionsToPackage{hypertexnames=false}{hyperref}
\documentclass{theoretics}

\usepackage{framed}

\usepackage[linesnumbered, noline, noend,noresetcount]{algorithm2e}
\newcommand{\cl}{\mathcal}
\newcommand{\sub}{\subseteq}
\DeclareMathOperator{\adj}{adj}
\DeclareMathOperator{\rank}{rank}
\newcommand{\set}[1]{\left\{#1\right\}}
\newcommand{\grp}[1]{\left(#1\right)}
\DeclareMathOperator{\head}{head}
\DeclareMathOperator{\tail}{tail}
\renewcommand{\sub}{\subseteq}
\newcommand{\dia}{\diamond}
\newcommand{\eout}{E_{\text{out}}}
\newcommand{\ein}{E_{\text{in}}}
\DeclareMathOperator{\poly}{poly}

\renewcommand{\b}{\boldsymbol}
\renewcommand{\t}{\text}
\newcommand{\new}{\text{new}}

\definecolor{darkorange}{rgb}{1.0, 0.55, 0.0} 

\definecolor{darkmidnightblue}{rgb}{0.0, 0.2, 0.4}
\definecolor{brilliantlavender}{rgb}{0.96, 0.73, 1.0}
\definecolor{alizarin}{rgb}{0.82, 0.1, 0.26}
\definecolor{byzantium}{rgb}{0.36, 0.22, 0.33}

\definecolor{brightube}{rgb}{0.82, 0.62, 0.91}
\usepackage{tikz}
\usetikzlibrary{backgrounds,decorations.pathreplacing,fadings}
\usepackage{float}

\usepackage[capitalise,noabbrev,nameinlink]{cleveref}

\makeatletter
\newcommand\IfRestateTF{%
  \ifx\label\thmt@gobble@label 
    \expandafter\@firstoftwo
  \else
    \expandafter\@secondoftwo
  \fi
}
\makeatother
\newcommand{\RestateRemark}{\IfRestateTF{{\normalfont\bfseries (Restated)}}{}}

\title{An Enumerative Perspective on Connectivity}

\ThCSshorttitle{An Enumerative Perspective on Connectivity}
\ThCSshortnames{S. Akmal}

\ThCSauthor[1]{Shyan Akmal}{naysh@mit.edu}[0000-0002-1825-0097]

\ThCSaffil[1]{MIT EECS and CSAIL, Cambridge MA, USA}

\ThCSthanks{This is an invited paper from SOSA 2024~\cite{Akmal24}. We thank Tahsin Saffat, Abhijit Mudigonda, and anonymous reviewers for extensive comments on earlier versions of this paper.
	Special thanks to Ce Jin for pointing out a mistake in the proof of \Cref{lem:lgv} in a previous version of this paper.}

\ThCSkeywords{maximum flow, all-pairs flow, connectivity, rank, algebraic graph algorithms, formal power series}

\ThCSyear{2024}
\ThCSarticlenum{26}
\ThCSreceived{Feb 2, 2024}
\ThCSrevised{Nov 11, 2024}
\ThCSaccepted{Nov 17, 2024}
\ThCSpublished{Dec 19, 2024}
\ThCSdoicreatedtrue

\begin{document}
	
	\maketitle
	
	\begin{abstract}
		Connectivity (or equivalently, unweighted maximum flow) is an important measure in graph theory and combinatorial optimization.
		Given a graph $G$ with vertices $s$ and $t$, the connectivity $\lambda(s,t)$ from $s$ to $t$ is defined to be the maximum number of edge-disjoint paths from $s$ to $t$ in $G$. 
		Much research has gone into designing fast algorithms for computing connectivities in graphs.
		Previous work showed that it is possible to compute connectivities for all pairs of vertices in directed graphs with $m$ edges in $\tilde{O}(m^\omega)$ time [Chueng, Lau, and Leung, FOCS 2011], where $\omega \in [2,2.3716)$ is the exponent of matrix multiplication.
		For the related problem of computing ``small connectivities,'' it was recently shown that for any positive integer $k$, we can compute 
		$\min(k,\lambda(s,t))$ for all pairs of vertices $(s,t)$ in a directed graph with $n$ nodes in $\tilde{O}((kn)^\omega)$ time [Akmal and Jin, ICALP 2023].
		
		In this paper, we present an alternate exposition of these $\tilde{O}(m^\omega)$ and $\tilde{O}((kn)^\omega)$ time algorithms, with simpler proofs of correctness.
		Earlier proofs were somewhat indirect, introducing an elegant but ad hoc ``flow vector framework'' for showing correctness of these algorithms.
		In contrast, we observe that these algorithms for computing exact and small connectivity values can be interpreted as testing whether certain generating functions enumerating families of edge-disjoint paths are nonzero. 
		This new perspective yields more transparent proofs, and ties the approach for these problems 
		more closely to the literature surrounding algebraic graph algorithms.
	\end{abstract}

	\section{Introduction}
	\label{sec:intro}
	
	Many problems in graph algorithms involve quantifying how ``connected'' different parts of a network are.
	One concept developed to study these questions is that of connectivity:
	given a graph $G$ with vertices $s$ and $t$, the \emph{connectivity} from $s$ to $t$, denoted by $\lambda(s,t)$, is the maximum number of edge-disjoint paths from $s$ to $t$ in $G$.
Connectivity is an old and well-studied measure in graph theory, important in computer science because the connectivity $\lambda(s,t)$ is equal to the \emph{maximum flow} value from $s$ to $t$ in the unweighted graph $G$.
Consequently, significant research has gone into designing fast algorithms for computing connectivities.

Suppose the graph $G$ has $n$ vertices and $m$ edges.
Computing an individual connectivity value in $G$ is easy: since the maximum flow between a fixed pair of nodes can be found in almost-linear time \cite{almost-linear-maxflow}, we know that for any fixed pair of vertices $(s,t)$ in $G$, we can compute $\lambda(s,t)$ in almost-optimal $m^{1 + o(1)}$ time.
For many applications however, knowing a single connectivity is not so useful, and it is instead far more informative to know the values of \emph{multiple} connectivities in a graph.

This motivates the \textsf{All-Pairs Connectivity (APC)} problem, where we are tasked with computing connectivities for all pairs of vertices in a given graph.
A long line of work recently culminated in a near-optimal $\tilde{O}(n^2)$ time algorithm for solving \textsf{APC} over \emph{undirected} graphs \cite{near-quadratic-gomory-hu}.
Throughout this paper, we focus on the general case where $G$ is directed.

\begin{framed} \noindent \textsf{All-Pairs Connectivity (APC)}\\
	Given a directed graph $G$, compute $\lambda(s,t)$ for all pairs of vertices $(s,t)$ in $G$. 
\end{framed}

We can of course solve \textsf{APC} naively in $n^2m^{1+o(1)}$ time, simply by solving a separate instance of maximum flow for each pair of vertices.
In dense graphs, this naive approach is actually the fastest known algorithm for \textsf{APC}!
In slightly sparse graphs however, we can do better and solve  \textsf{APC} in $\tilde{O}(m^\omega)$ time \cite{CheungLauLeung2013}, where $\omega$ is the exponent of matrix multiplication (i.e., $\omega$ is the smallest positive real such that two $a\times a$ matrices can be multiplied using $a^{\omega+o(1)}$ arithmetic operations).
The current fastest algorithms for matrix multiplication imply that $\omega < 2.3716$ \cite{mmult-life}.
If $\omega = 2$, then the $\tilde{O}(m^\omega)$ time algorithm is always at least as fast as the naive approach.
The lack of progress in finding faster algorithms for \textsf{APC} has motivated researchers to consider relaxations of \textsf{APC}, including the \textsf{$k$-Bounded All-Pairs Connectivity ($k$-APC)} problem:
\begin{framed} \noindent \textsf{$k$-Bounded All-Pairs Connectivity Problem ($k$-APC)}\\
	Given a directed graph $G$, compute $\min(k,\lambda(s,t))$ for all pairs of vertices $(s,t)$ in $G$. 
\end{framed}

The \textsf{$k$-APC} problem is relevant in contexts where knowing the precise connectivity values between ``well-connected'' nodes is not important, and  instead we care more  about distinguishing for each pair of vertices whether its connectivity is small or large (where $k$ is our cutoff for what counts as ``small'' and ``large'').
Since the connectivity between any pair of nodes in $G$ is at most $n-1$, the \textsf{$k$-APC} problem
recovers the general \textsf{APC} problem when $k=n-1$.
As $k$ gets smaller, \textsf{$k$-APC} 
intuitively becomes easier. 
Indeed, it was recently shown that for any positive integer $k$, \textsf{$k$-APC} can be solved in $\tilde{O}((kn)^\omega)$ time \cite{k-APC}.

In summary, the following results are known for the \textsf{APC} and \textsf{$k$-APC} problems.

\begin{restatable}{theorem}{apc}\RestateRemark
	\label{apc}
	There is an algorithm solving \textsf{APC} in $\tilde{O}(m^\omega)$ time.
\end{restatable}
\begin{restatable}{theorem}{kapc}\RestateRemark
	\label{kapc}
	There is an algorithm solving \textsf{$k$-APC} in $\tilde{O}((kn)^\omega)$ time.
\end{restatable}

The algorithms establishing \Cref{apc,kapc} are straightforward.
However, the proofs of correctness for these algorithms presented in the literature are not at all obvious, 
and involve
arguments in a somewhat complicated ``flow vector framework.'' 

In this note, we present simpler proofs of correctness for the \textsf{APC} and \textsf{$k$-APC} algorithms from \cite{CheungLauLeung2013} and \cite{k-APC} respectively.

\subsubsection*{Comparison With Previous Work}

The algorithms for \textsf{APC} and \textsf{$k$-APC} from \cite{CheungLauLeung2013} and \cite{k-APC} work in similar ways.
Each algorithm first constructs a certain random matrix $M$ whose rows and columns are indexed by edges of $G$.
Then, for each pair of vertices $(s,t)$, the algorithms return the value of 
\[\rank~M[\eout(s),\ein(t)]\]
as the answer for that pair, where $\eout(s)$ is the set of edges exiting $s$ and $\ein(t)$ is the set of edges entering $t$.
To prove correctness, one simply needs to show that with high probability, for every pair of vertices $(s,t)$ the rank expression above equals $\lambda(s,t)$  or $\min(k,\lambda(s,t))$, depending on whether $M$ was designed to solve \textsf{APC} or \textsf{$k$-APC} respectively.

Previous proofs of correctness for these algorithms use the flow vector framework.
In this framework, we fix a source node $s$, and 
imagine pumping out random vectors along the edges exiting $s$.
Intuitively, we let these vectors propagate throughout the graph and use them to assign vectors to each edge of $G$ in a manner that satisfies certain ``flow conservation'' rules.
One can then argue that for any vertex $s$ and edge $e$, the column vector $M[\eout(s), e]$ equals the vector assigned to edge $e$ in $G$.
This interpretation of the entries of $M$ then allows one to establish the connection between ranks of submatrices of $M$ and connectivity.

See \cite[Section 3]{k-APC} for a more detailed overview of the approach in previous work.

In this work, we present alternate proofs of correctness for these \textsf{APC} and \textsf{$k$-APC} algorithms.
Our proofs 
provide a combinatorial interpretation of determinants of submatrices of $M$ as \emph{generating functions} enumerating families of edge-disjoint walks in $G$.
This approach lets us directly connect 
$\rank M[\eout(s),\ein(t)]$ to $\lambda(s,t)$, without introducing any of the auxiliary scaffolding of the flow vector framework.

This new perspective yields simpler and more direct proofs of correctness for the \textsf{APC} and \textsf{$k$-APC} algorithms than what was previously presented in the literature.
For example, our proof does not even use Menger's theorem (the fact that $\lambda(s,t)$ is equal to the minimum number of edge deletions needed to disconnect $t$ from $s$), which previous proofs relied on.
From a pedagogical perspective, the new proofs are more transparent, making it clear how and \emph{why} matrix rank relates to connectivities, so that the exposition of the  
\textsf{APC} and \textsf{$k$-APC} algorithms in this paper should be easier to teach and motivate compared to previous work.

Finally, our generating function approach also yields somewhat more expressive results than the flow vector framework, leading to 
an easier proof of correctness for the \textsf{$k$-APC} algorithm in particular.
Previously in \cite[Section 4]{k-APC}, to prove correctness of the \textsf{$k$-APC} algorithm the authors had to manually reprove ``low-rank'' versions of all the flow vector framework results shown by \cite{CheungLauLeung2013} for the general \textsf{APC} problem.
This is not necessary in our approach: once we establish our generating function for edge-disjoint walks in $G$, some small additional reasoning yields both the \textsf{APC} and \textsf{$k$-APC} algorithms.

The only technical wrinkle in our approach is that our combinatorial view of the problem involves manipulating formal power series.
However, as we discuss later in \Cref{remark:DAGs}, even this ingredient can be removed if one only wishes to establish \Cref{apc,kapc} over directed acyclic graphs (which is already an interesting result, perhaps more suitable for teaching these algorithms in a classroom setting).

\subsection*{Organization}
In \Cref{sec:prelim} we identify notation and assumptions used throughout the paper.
In \Cref{sec:formal-power-series} we review standard definitions  and properties of formal power series.
In \Cref{sec:enumeration}, we construct a matrix of formal power series whose entries enumerate families of edge-disjoint walks in a graph.
In \Cref{sec:APC,sec:k-APC} we leverage the construction from \Cref{sec:enumeration} to give simple proofs of \Cref{apc,kapc} respectively. 
We conclude in \Cref{sec:conclusion} by mentioning some connections between our arguments and classical results in combinatorics and computer science, and highlighting open problems related to computing connectivities.

\section{Preliminaries}
\label{sec:prelim}

\paragraph*{General Notation}

Given a positive integer $a$, we let $[a] = \set{1, \dots, a}$ denote the set of the first $a$ consecutive positive integers.

\paragraph*{Graph Assumptions and Notation}

Throughout, we let $G$ denote the input graph, with $n$ nodes and $m$ edges.
We let $V$ and $E$ denote the vertex and edge sets of $G$ respectively.
We assume that $G$ is weakly connected (i.e., the underlying undirected graph of $G$ is connected), so that $n-1\le m$.
This is without loss of generality: if $G$ is disconnected, we can solve \textsf{APC} and \textsf{$k$-APC} on $G$ by solving these problems separately on each weakly connected component of $G$.
We assume that $G$ is simple (i.e., $G$ does not have self-loops or parallel edges between nodes).

For a directed edge $e = (u,v)\in E$ going from $u$ to $v$, we let $\tail(e) = u$ and $\head(e) = v$ denote the vertices that $e$ exits and enters respectively. 

A \emph{walk} in $G$ is a sequence of edges $W = \langle e_1, \dots, e_\ell\rangle$ such that $\head(e_j) = \tail(e_{j+1})$ for each $j\in[\ell-1]$.
We say $W$ is a walk starting at $e_1$ and ending at $e_\ell$.
We say $W$ is a \emph{path} if no two of its edges enter the same vertex, and its starting $s=\tail(e_1)$ and ending $t=\head(e_\ell)$ vertices are distinct.
Such a path from $s$ to $t$ is referred to as an $st$-path.

\paragraph*{Finite Field Computation}
Throughout, we work over a finite field $\mathbb{F} = \mathbb{F}_{2^q}$ of characteristic two.
We set $q = \Theta(\log n)$ large enough so that the field has size $2^q\ge 12n^6$ (looking ahead, this is to ensure that our algorithms work with high probability).
We can perform arithmetic operations over this field in $q^{1+o(1)} = \poly(\log n)$ time.

\paragraph*{Matrix Notation}

Given a matrix $M$, for any row index $i$ and column index $j$ we let $M[i,j]$ be the $(i,j)$ entry of $M$.
Given subsets $I$ and $J$ of row and column indices respectively, we let $M[I,J]$ be the submatrix of $M$ restricted to rows in $I$ and columns in $J$.
We also let $M[I,\cdot]$ be the submatrix restricted to rows in $I$ and all columns, and $M[\cdot,J]$ be the submatrix on all rows and restricted to columns in $J$.
We let $\rank M$ denote the rank of $M$, defined to be largest nonnegative integer $r$ such that $M$ contains an $r\times r$ submatrix with nonzero determinant.
When $M$ is a square matrix, we let $\det M$ denote the determinant of $M$, $\adj(M)$ denote the adjugate of $M$, and $M^{-1}$ denote the inverse of $M$ (if it exists).

\paragraph*{Matrix Computation}

We recall the following results concerning matrix computation.

\begin{proposition}[Matrix Inversion]
	\label{matrix:inverse}
	For any positive integer $a$, we can compute the inverse of an $a\times a$ matrix over a field in $O(a^\omega)$ field operations.
\end{proposition}

\begin{proposition}[Matrix Rank]
	\label{matrix:rank}
	For any positive integers $a$ and $b$, we can compute the rank of an $a\times b$ matrix over a field in $O(ab^{\omega-1})$ field operations.
\end{proposition}

Proofs of \Cref{matrix:inverse,matrix:rank} can be found in \cite{lecture-notes-inversion} and 
\cite{matrix-rank} respectively. 

\paragraph*{Identity Testing}

To prove correctness of the \textsf{APC} and \textsf{$k$-APC} algorithms, we  use the fact that random evaluations of a low degree polynomial (or more generally, rational function whose numerator and denominator have low degree) over a large field are nonzero with high probability.

\begin{proposition}
	\label{prop:ore}
	Let $P$ be a nonzero $r$-variate polynomial of degree at most $d$.
	Then a uniform random evaluation of $P$ over $\mathbb{F}^r$ is nonzero with probability at least $1-d/|\mathbb{F}|$.
\end{proposition}

For an accessible proof of \Cref{prop:ore}, see \cite[Theorem 7.2]{Motwani1995}.

\begin{corollary}[Rational Identity Testing]
	\label{cor:rat-test}
	Let $R = P/Q$ be a rational function, represented as the ratio of two nonzero polynomials $P$ and $Q$.
	Suppose $P$ and $Q$ each have degree at most $d$.
	If we assign each variable of $R$ an independent, uniform random element of $\mathbb{F}$, then $R$ has nonzero evaluation with probability at least $1-2d/|\mathbb{F}|$.
\end{corollary}
\begin{proof}
	Under random evaluation over $\mathbb{F}$, by \Cref{prop:ore} and the union bound, $P$ and $Q$ are both nonzero with probability at least $1- 2d/|\mathbb{F}|$. 
	So with this probability, the rational function $R=P/Q$ also has nonzero evaluation, as claimed.
\end{proof}

\section{Power Series Preliminaries}
\label{sec:formal-power-series}
Our algorithms for computing connectivities work by constructing generating functions for families of edge-disjoint walks.
These generating functions involve infinite sums, so in this section we review properties of formal power series, a generalization of polynomials which allow for infinite sums.
The results we review are simple, and mostly involve observing that basic facts which hold for polynomials still hold in the infinite case of formal series.

We also note that if one is interested in solving \textsf{APC} and \textsf{$k$-APC} only in the special case of directed acyclic graphs, then it suffices to work with polynomials (no formal power series are needed) and this section can be skipped.
We discuss this simplification for acyclic graphs in detail later on in  \Cref{remark:DAGs}.

Fix a finite set $J$, and consider the set of variables $\set{x_j}_{j\in J}$ indexed by $J$.
A polynomial is a finite linear combination of products of these variables.
A formal power series is simply a generalization of polynomials which allows for infinite sums.

Let $\mathbb{N}^J$ be the set of all sequences of nonnegative integers indexed by $J$.
Given $\boldsymbol{d}\in \mathbb{N}^J$, we let $d_j$ denote the $j^{\text{th}}$ element of $\b{d}$ for each $j\in J$.
Then a \emph{formal power series} $F$ is identified by a sequence of coefficients $a_{\b{d}}$ in $\mathbb{F}$, one for each $\b{d}\in \mathbb{N}^J$, and we write 
\[F = \sum_{\b{d}\in\mathbb{N}^J} a_{\b{d}}\prod_{j\in J} x_j^{d_j}.\]
We let $\b{0}$ denote the all-zeros sequence in $\mathbb{N}^J$, and say $a_{\b{0}}$ is the \emph{constant term} of $F$.
In general, given $\b{d}\in\mathbb{N}^J$, the monomial corresponding to $\b{d}$ in $F$ (if it appears with nonzero coefficient $a_{\b{d}}\neq 0$) is said to have degree 
\[\sum_{j\in J} d_j.\]
Given formal series
\[F = \sum_{\b{d}\in\mathbb{N}^J} a_{\b{d}}\prod_{j\in J} x_j^{d_j}\quad\text{and}\quad
H = \sum_{\b{d}\in\mathbb{N}^J} b_{\b{d}}\prod_{j\in J} x_j^{d_j}\]
we define their sum 
\[F + H = \sum_{\b{d}\in\mathbb{N}^J} \grp{a_{\b{d}} + b_{\b{d}}}\prod_{j\in J} x_j^{d_j}\]
and product 
\begin{equation}
	\label{eq:product}
	F\cdot H = \sum_{\b{d}\in\mathbb{N}^J}
	\grp{\sum_{\substack{\b{d_1},\b{d_2}\in\mathbb{N}^J \\ \b{d_1} + \b{d_2} = \b{d}}} a_{\b{d_1}}b_{\b{d_2}} }
	\prod_{j\in J} x_j^{d_j}
\end{equation}
in the natural way, generalizing arithmetic over polynomials.
These operations make the set of polynomials over $\mathbb{F}$ a subring of the ring of formal power series (where the additive and multiplicative identities are the constant polynomials $0$ and $1$ respectively).

\begin{proposition}[Power Series Inversion]
	\label{prop:inverse}
	If $F$ is a formal power series with constant term $1$, then there is a unique formal series $H=F^{-1}$ such that $F\cdot H = 1$.
	Moreover, the constant term of $H$ equals $1$.
\end{proposition}
\begin{proof}
	Suppose 
	\[F = \sum_{\b{d}\in\mathbb{N}^J} a_{\b{d}}\prod_{j\in J} x_j^{d_j}.\]
	We define the sequence $b_{\b{d}}$ of coefficients in $\mathbb{F}$ for all $\b{d}\in\mathbb{N}^J$ inductively, by setting $b_{\b{0}} = 1$, and taking
	\begin{equation}
		\label{eq:induct}
		b_{\b{d}} = -\grp{
			\sum_{\b{d'}\prec \,\b{d}} 
			a_{\b{d}-\b{d'}}b_{\b{d'}}
		}
	\end{equation}
	where $\b{d}'\prec \b{d}$ means that  $\b{d'}\in\mathbb{N}^J$ is componentwise less than or equal to $\b{d}$, and $\b{d'}\neq \b{d}$.
	
	Then if we set 
	\[H = \sum_{\b{d}\in\mathbb{N}^J} b_{\b{d}}\prod_{j\in J} x_j^{d_j}\]
	it follows from the definition of multiplication in \cref{eq:product}, the relationship from \cref{eq:induct}, and the fact that $a_{\b{0}} = b_{\b{0}} = 1$, that we have 
	\[F\cdot H = 1.\]
	This inverse $H$ is unique, because if another formal series $H'$ satisfies $F\cdot H' = 1$, then the constant term of $H'$ is $1$ since the product of the constant terms of $F$ and $H'$ are $1$, and 
	\[H' = H'\cdot 1 = H'\cdot (F\cdot H) = (H'\cdot F)\cdot H = 1\cdot H = H.\]
	Thus $F$ has a unique multiplicative inverse.
\end{proof}

\paragraph*{Matrices of Formal Series}

In this paper, we work with matrices whose entries are formal power series.
Such matrices naturally arise when computing inverses of polynomial matrices.

\begin{proposition}[Geometric Series Formula]
	\label{prop:geo-series}
	Suppose $X$ is a square matrix with polynomial entries such that every entry of $X$ has constant term equal to zero.
	Then 
	\begin{equation}
		\label{eq:geo-series}
		(I-X)^{-1} = \sum_{\ell=0}^\infty X^\ell.
	\end{equation}
\end{proposition}
\begin{proof}
	Since every entry of $X$ has constant term zero, every nonzero entry of $X^\ell$ has degree at least $\ell$. 
	Consequently, the infinite sum from the right-hand side of \cref{eq:geo-series} is well-defined, because for any $\b{d}\in\mathbb{N}^J$, only finitely many terms contribute to the coefficient of 
	\[\prod_{j\in J}x_j^{d_j}\]
	in each entry of the sum.
	It suffices to prove that the product 
	\begin{equation}
		\label{eq:geo-product}
		(I-X)\left(\sum_{\ell=0}^\infty X^\ell\right)
	\end{equation}
	is equal to the identity matrix.
	
	For any fixed integer $d \ge 0$, let $M_d$ be the matrix from \cref{eq:geo-product} with entries restricted to terms of degree at most $d$.
	Then since nonzero entries of $X^\ell$ have degree at least $\ell$, we see that $M_d$ is equal to the matrix
	\[(I-X)\grp{\sum_{\ell=0}^d X^\ell} = \sum_{\ell=0}^d\grp{X^\ell - X^{\ell+1}} = I - X^{d+1}\]
	with entries restricted to terms of degree at most $d$.
	Since every nonzero entry of $X^{d+1}$ has degree greater than $d$, we see that $M_d = I$ is the identity matrix.
	Since this equation holds for every fixed $d\ge 0$, the product in \cref{eq:geo-product} equals the identity matrix as claimed. 
\end{proof}

\section{Enumeration with Edge-Adjacency Matrices}
\label{sec:enumeration}

In this section we
construct a   
matrix $\Gamma$ of formal power series, whose 
rows and columns are indexed by edges of $G$.
The matrix $\Gamma$ will be designed to have the special property that for any equal-size subsets of edges $S,T\sub E$, the determinant
\[\det~\Gamma[S,T]\]
is nonzero as a formal series if and only if there are edge-disjoint paths connecting the edges in $S$ to the edges in $T$.
Since connectivity is defined in terms of edge-disjoint paths, 
intuitively our
construction of  $\Gamma$ will let us solve the \textsf{APC} and \textsf{$k$-APC} problems in \Cref{sec:APC,sec:k-APC} by performing certain matrix computations.

For every pair of edges $(e,f)$ in $G$ such that $\head(e) = \tail(f)$ (i.e., edge $e$ enters the vertex that edge $f$ exits), we introduce an indeterminate variable $x_{ef}$.
Let $X$ be the $m\times m$ matrix with rows and columns indexed by edges of $G$, such that for each pair of edges $(e,f)$ in $G$, we have 
\begin{equation}
	\label{eq:X-def}
	X[e,f] = \begin{cases}
		x_{ef} & \text{if } \head(e)=\tail(f) \\
		0 & \text{otherwise.}
	\end{cases}
\end{equation}

We enumerate walks not by counting their number, but by assigning each walk a monomial weight, that records information about the edges traversed in the walk.
Enumeration for our purposes corresponds to summing the weights of all walks (or collections of walks) in a certain family of interest.

Given a walk $W = \langle e_1, \dots, e_{\ell}\rangle$, viewed as a sequence of edges $e_j$, we let the weight
\[\xi(W) = \prod_{j=1}^{\ell-1}x_{e_je_{j+1}}\]
of $W$ be the monomial $\xi(W)$ recording all pairs of consecutive edges traversed by $W$.
By convention, we assign a walk $W$ of length one (i.e., a single edge) weight $\xi(W) = 1$.
More generally, given a collection of walks $\cl{C} = \langle W_1, \dots, W_r\rangle$ we let the weight 
\[\xi(\cl{C}) = \prod_{j=1}^r \xi(W_j)\]
of $\cl{C}$ be the product of the weights of the individual walks.

\subsubsection*{Enumerating Walks}

Given edges $e,f\in E$ and an integer $\ell\ge 1$, let $\cl{W}_\ell(e,f)$ denote the set of all walks beginning at $e$ and ending at $f$ of length $\ell$.
One way of interpreting the definition of $X$ from \cref{eq:X-def} is that the $(e,f)$ entry of $X$ enumerates all walks of length two from $e$ to $f$ in $G$.
These are precisely the walks in $\cl{W}_2(e,f)$.
The next result observes that higher powers of $X$ enumerate walks of longer lengths in $G$.

\begin{proposition}
	\label{prop:power-walk}
	For any edges $e,f\in E$ and integer $\ell \ge 0$, we have 
	\[X^\ell[e,f] = \sum_{W\in \cl{W}_{\ell+1}(e,f)} \xi(W).\]
\end{proposition}
\begin{proof}
	By expanding out the definition of matrix multiplication, we see that 
	\[X^\ell[e,f] = \sum_{\substack{e_0, \dots, e_{\ell}\in E \\ e_0 = e \\ e_{\ell} = f}} \prod_{j=0}^{\ell-1} X[e_j,e_{j+1}].\]
	By definition, $X[e_j,e_{j+1}] = x_{e_je_{j+1}}$ if we can step from $e_j$ to $e_{j+1}$ in $G$, and is zero otherwise.
	Thus, the product 
	\[\prod_{j=0}^{\ell-1} X[e_j,e_{j+1}]\]
	is nonzero if and only if $W = \langle e_0, \dots, e_{\ell}\rangle$ is a walk of length $(\ell+1)$ in $G$.
	In this case, 
	\[\prod_{j=0}^{\ell-1} X[e_j,e_{j+1}] = \prod_{j=0}^{\ell-1} x_{e_je_{j+1}} = \xi(W)\]
	so 
	we have
	\[X^\ell[e,f] = \sum_{W\in\cl{W}_{\ell+1}(e,f)} \xi(W)\]
	as claimed.
\end{proof}

\begin{corollary}[Enumerating Walks]
	\label{cor:inverse-walks}
	For any edges $e,f\in E$, we have
	\[(I-X)^{-1}[e,f] = \sum_{\ell=0}^\infty \grp{\sum_{W\in\cl{W}_{\ell+1}(e,f)}\xi(W)}.\]
\end{corollary}
\begin{proof}
	This result follows by combining the geometric series formula from \Cref{prop:geo-series} with the enumerative property of powers of $X$ from \Cref{prop:power-walk}.
\end{proof}

\subsubsection*{Enumerating Edge-Disjoint Walks}

Given subsets of edges $S,T\sub E$ of equal size $|S|=|T|=r\ge 1$ and an integer $\ell\ge 1$, we define $\cl{F}_\ell(S,T)$ to be the family of collections of $r$ walks of total length $\ell$, beginning at different edges of $S$ and ending at different edges of $T$.
If we fix some ordering $e_1, \dots, e_r$ of the edges in $S$, then we can view each element of $\cl{F}_{\ell}(S,T)$ as a sequence of walks $\langle W_1, \dots, W_r\rangle$ satisfying the properties that 
each $W_i$ begins at $e_i$ and ends at some edge of $T$, the $W_i$ walks all end at distinct edges of $T$, and the sum of the lengths of the $W_i$ walks is $\ell$.

Furthermore, let $\cl{D}_\ell(S,T)\sub \cl{F}_\ell(S,T)$
be the family of collections of $r$ \emph{edge-disjoint} walks from $S$ to $T$ of total length $\ell$.

\Cref{cor:inverse-walks} shows that entries of $\Gamma = (I-X)^{-1}$ enumerate walks in $G$.
The following result uses this fact to show that determinants of submatrices of $\Gamma$ enumerate collections of walks in $G$, beginning and ending at different edges.
This is a simple observation, following immediately from the definition of the determinant.
The proof appears somewhat long only because we spell out the details of each step in the calculation.

\begin{lemma}[Arbitrary Walks]
	\label{lem:det-definition}
	For any equal-size subsets of edges $S,T\sub E$, we have 
	\[\det~(I-X)^{-1}[S,T] = \sum_{\ell=1}^{\infty} \grp{\sum_{\cl{C}\in \cl{F}_\ell(S,T)}\xi(\cl{C})}\]
\end{lemma}
\begin{proof}
	For convenience, write $\Gamma = (I-X)^{-1}$.
	
	Let $\mathfrak{S}(S,T)$ be the set of all bijections from $S$ to $T$.
	
	By the definition of the determinant, we have 
	\begin{equation}
		\label{eq:Laplace}
		\det \Gamma[S,T] = \sum_{\pi\in\mathfrak{S}(S,T)} \prod_{e\in S} \Gamma[e,\pi(e)].
	\end{equation}
	Note that we do not include a factor  for the sign of $\pi$ in the above equation, because we are working over a field of characteristic two.
	
	By \Cref{cor:inverse-walks}, for each $e\in S$ we have 
	\begin{equation}
		\label{eq:single-det}
		\Gamma[e,\pi(e)] = \sum_{\ell=0}^\infty \grp{\sum_{W\in\cl{W}_{\ell+1}(e,\pi(e))}\xi(W)}.
	\end{equation}
	
	Write $S = \set{e_1, \dots, e_r}$, where $r= |S|=|T|$.
	
	By multiplying the above equation over all choices of $e\in S$, we have 
	\begin{equation}
		\label{eq:one-product}
		\prod_{e\in S}\Gamma[e,\pi(e)] = \prod_{i=1}^r \grp{\sum_{\ell=0}^\infty \grp{\sum_{W\in\cl{W}_{\ell+1}(e_i,\pi(e_i))}\xi(W)}}.
	\end{equation}
	Now, let $\cl{L}$ be the set of all $r$-tuples $(\ell_1,\dots,\ell_r)$ of positive integers summing to
	\[\ell_1 + \dots +\ell_r = \ell.\]
	If we expand out the product on the right-hand side of \cref{eq:one-product} and group terms according to the total length of the walks they come from, we obtain
	\[\prod_{i=1}^r \grp{\sum_{\ell=0}^\infty \grp{\sum_{W\in\cl{W}_{\ell+1}(e_i,\pi(e_i))}\xi(W)}} = \sum_{\ell=1}^{\infty} \grp{
		\sum_{\substack{(\ell_1,\dots,\ell_r)\in\cl{L} \\ W_i\in \cl{W}_{\ell_i}(e_i,\pi(e_i))}} 
		\prod_{i=1}^r\xi(W_i)}.\]
	To clarify the expression above: in the right-hand side of the above equation, the first inner summation is over all choices of positive integers $\ell_1,\dots, \ell_r$ which sum to $\ell$, and choices of walks $W_1,\dots, W_r$ where $W_i$ is a walk of length $\ell_i$ from $e_i$ to $\pi(e_i)$.
	This is simply the result of distributing the product over $i\in[r]$ on the left-hand side of the equation over the sum of  walks of all possible lengths from $e_i$ to $\pi(e_i)$.
	
	By chaining the above equation together with \cref{eq:Laplace,eq:single-det,eq:one-product}, and interchanging summation, we get that 
	
	\begin{equation}
		\label{eq:nearly-there}
		\det~\Gamma[S,T] = \sum_{\ell=1}^\infty \grp{\sum_{\pi\in\mathfrak{S}(S,T)} \sum_{\substack{\\(\ell_1,\dots,\ell_r)\in\cl{L} \\ W_i\in \cl{W}_{\ell_i}(e_i,\pi(e_i))}} 
			\prod_{i=1}^r\xi(W_i)}.
	\end{equation}
	
	To simplify \cref{eq:nearly-there}, observe that 
	for any choice of bijection $\pi\in\mathfrak{S}(S,T)$, lengths $(\ell_1,\dots,\ell_r)\in\cl{L}$, and  walks $W_i\in\cl{W}_{\ell_i}(e_i,\pi(e_i))$, the collection $\langle W_1, \dots, W_r\rangle$ is a sequence of walks from $S$ to $T$ of total length $\ell$.
	Conversely, any collection $\cl{C}\in\cl{F}_\ell(S,T)$  has walks whose lengths sum up to $\ell$, and corresponds to a unique bijection $\pi\in\mathfrak{S}(S,T)$, obtained by checking which starting edges in $S$ are connected to which ending edges in $T$ by walks in $\cl{C}$.
	
	Thus, the inner nested summation above is equivalent to a single sum over all collections of walks in $\cl{F}_\ell(S,T)$.
	Since the weight of a collection $\cl{C} = \langle W_1, \dots, W_r\rangle$
	is precisely
	\[\xi(\cl{C}) = \prod_{i=1}^r \xi(W_i),\]
	the discussion from the previous paragraph together with \Cref{eq:nearly-there} implies that 
	
	\[\det\Gamma[S,T] = \sum_{\ell=1}^{\infty} \grp{\sum_{\cl{C}\in \cl{F}_\ell(S,T)}\xi(\cl{C})}\]
	which proves the desired result.
\end{proof}

We now observe that the determinant sieves out collections of intersecting walks, so that only edge-disjoint families of walks are included in its enumeration.

\begin{lemma}[Intersecting Walks Cancel]
	\label{lem:lgv}
	For any equal-size subsets of edges $S,T\sub E$ and integer $\ell\ge 1$, we have 
	\[\sum_{\cl{C}\in \cl{F}_\ell(S,T)} \xi(\cl{C})= \sum_{\cl{C}\in \cl{D}_\ell(S,T)} \xi(\cl{C}).\]
\end{lemma}
\begin{proof}
	Fix $S,T\sub E$ and integer $\ell\ge 1$.
	Let $r = |S| = |T|$.
	
	For convenience, abbreviate $\cl{F} = \cl{F}_\ell(S,T)$ and $\cl{D} = \cl{D}_\ell(S,T)$.
	Let $\cl{S} = \cl{F}\setminus\cl{D}$ be the family of all collections of $r$ walks beginning at different edges of $S$ and ending at different edges of $T$, such that at least two walks in the collection intersect at an edge. 
	By definition we have 
	\[\sum_{\cl{C}\in \cl{F}} \xi(\cl{C}) = \sum_{\cl{C}\in \cl{D}} \xi(\cl{C}) + \sum_{\cl{C}\in \cl{S}} \xi(\cl{C}).\]
	So to prove the claim, it suffices to show that 
	\[\sum_{\cl{C}\in \cl{S}} \xi(\cl{C})\]
	is the zero polynomial.
	We prove this by pairing up collections $\cl{C}$ in $\cl{S}$ of equal weight $\xi(\cl{C})$, and observing that contributions from such collections vanish modulo two.

	\begin{figure}
		\centering
		\begin{tikzpicture}[scale=1.5,
vtx/.style={
inner sep=0pt},
topedge/.style={-stealth, thick, darkorange},
botedge/.style={-stealth, thick, darkmidnightblue, densely dashed}]
    \def\top{0.7cm};
    \def\half{0cm};
    \def\bottom{-0.7cm};

    \def\stp{1cm};

    \def\extr{0.7cm};
    \def\extrw{1.4cm};

    \node[vtx] (1) at (\stp, \top) {};
    \node[vtx] (2) at (2*\stp, \top) {};
    \node[vtx] (3) at (3*\stp, \half) {};
    \node[vtx] (4) at (4*\stp, \half) {};
    \node[vtx] (5) at (5*\stp, \top) {};
        \node[vtx] (6) at (6*\stp, \top) {};

    \draw[topedge] (1) -- (2) node [midway, above] {$e_1$};
    \draw[topedge] (2) -- (3) node [near end, above=5] {$e_2$};
        \draw[topedge] (3) -- (4) node[midway, above] {$\textcolor{byzantium}{g}$};
    \draw[topedge] (4) -- (5) node [near start, above=5] {$e_3$};
        \draw[topedge] (5) -- (6) node [midway, above] {$e_4$};

    \node[vtx] (7) at (\stp, \bottom) {};
    \node[vtx] (8) at (2*\stp, \bottom) {};
        \node[vtx] (9) at (5*\stp, \bottom) {};
    \node[vtx] (10) at (6*\stp, \bottom) {};

    \draw[botedge] (7) -- (8) node [midway, above] {$f_1$};
    \draw[botedge] (8) -- (3) node [near end, below=2.5] {$f_2$};
            \draw[botedge] (3) -- (4);
    \draw[botedge] (4) -- (9) node [near start, below=2.5] {$f_3$};
        \draw[botedge] (9) -- (10) node [midway, above] {$f_4$};

    \node (100) at (6*\stp + \extr, \half) {};
        \node (101) at (6*\stp + \extr + \extrw, \half) {};

    \draw[latex-latex, ultra thick, alizarin] (100) -- (101);

    \def\halfshift{6*\stp + 2*\extr + \extrw};

        \node[vtx] (1s) at (\halfshift, \top) {};
        \node[vtx] (2s) at (\halfshift + \stp, \top) {};
    \node[vtx] (3s) at (\halfshift + 2*\stp, \half) {};
    \node[vtx] (4s) at (\halfshift + 3*\stp, \half) {};
    \node[vtx] (5s) at (\halfshift + 4*\stp, \top) {};
    \node[vtx] (6s) at (\halfshift + 5*\stp, \top) {};

    \node[vtx] (7s) at (\halfshift, \bottom) {};
    \node[vtx] (8s) at (\halfshift + \stp, \bottom) {};
        \node[vtx] (9s) at (\halfshift + 4*\stp, \bottom) {};
    \node[vtx] (10s) at (\halfshift + 5*\stp, \bottom) {};

    \draw[topedge] (1s) -- (2s) node [midway, above] {$e_1$};
    \draw[topedge] (2s) -- (3s) node[near end, above=5] {$e_2$};
    \draw[topedge] (3s) -- (4s) node[midway,above] {$\textcolor{byzantium}{g}$};
    \draw[topedge] (4s) -- (9s) node[near start, below=2.5] {$f_3$};
    \draw[topedge] (9s) -- (10s) node[midway, above] {$f_4$};

    \draw[botedge] (7s) -- (8s) node[midway, above] {$f_1$};
   \draw[botedge] (8s) -- (3s) node[near end, below=2.5] {$f_2$};
   \draw[botedge] (3s) -- (4s);
   \draw[botedge] (4s) -- (5s) node[near start, above=5] {$e_3$};
   \draw[botedge] (5s) -- (6s) node[midway,above] {$e_4$};
\end{tikzpicture}
		\caption{Given walks $W_i = \langle e_1, e_2, g, e_3, e_4\rangle$ and $W_j = \langle f_1, f_2, g, f_3, f_4\rangle$ overlapping at $g$, we can swap their suffixes to produce walks $W'_i = \langle e_1, e_2, g, f_3, f_4\rangle$ and $W'_j = \langle f_1, f_2, g, e_3, e_4\rangle$ which still overlap at $g$.
			The weight  $\xi(W_i, W_j) = (x_{e_1e_2}x_{e_2g}x_{ge_3}x_{e_3e_4})\cdot (x_{f_1f_2}x_{f_2g}x_{gf_3}x_{f_3f_4})$ of the first pair $(W_i, W_j)$ is precisely equal to the weight $\xi(W'_i,W'_j) = (x_{e_1e_2}x_{e_2g}x_{gf_3}x_{f_3f_4})\cdot (x_{f_1f_2}x_{f_2g}x_{ge_3}x_{e_3e_4})$ of the second pair $(W'_i, W'_j)$, because each pair traverses the same multiset of consecutive pairs of edges.}
		\label{fig:swap}
	\end{figure}
	
	Fix an ordering $e_1, \dots, e_r$ of the edges in $S$.
	Take any $\cl{C} = \langle W_1, \dots, W_r\rangle\in\cl{S}$, with the walks ordered so that $W_i$ begins at edge $e_i$.
	By assumption, at least two walks in $\cl{C}$ overlap at an edge.
	Let $i\in[r]$ be the smallest index such that $W_i$ intersects some other walk in $\cl{C}$ at an edge.
	Let $e$ be the first edge in $W_i$ which is contained in
	another walk of $\cl{C}$.
	Let $j\in[r]$ be the smallest index $j>i$ such that $W_j$ overlaps with $W_i$ at edge $e$. 
	
	We can split the walk $W_i$ uniquely 
	\[W_i = A_i\dia B_i\]
	as the concatenation of a prefix walk $A_i$ not including edge $e$, and a suffix walk $B_i$ beginning with edge $e$.
	We can similarly split $W_j$ uniquely 
	\[W_j = A_j\dia B_j\]
	as the concatenation of a prefix $A_j$ not including $e$, and a suffix $B_j$ beginning with $e$.
	
	Now, define walks
	\[W'_i = A_i\dia B_j\quad\text{and}\quad W'_j = A_j\dia B_i\]
	by swapping the suffixes of $W_i$ and $W_j$.
	An example of this operation is depicted in \Cref{fig:swap}.
	For all $l\in [r]$ with $l\not\in\set{i,j}$,
	set $W'_l = W_l$.
	Define a new collection  of walks 
	\[\cl{C}'=\langle W_1', \dots, W'_r\rangle\] 
	by replacing $W_i$ and $W_j$ in $\cl{C}$ with $W'_i$ and $W'_j$ respectively.
	
	Since $W_i$ and $W_j$ end at different edges of $T$, we know that $W'_i\neq W_i$ and $W'_j\neq W_j$.
	This shows that $\cl{C}'\neq \cl{C}$.
	Since the walks in $\cl{C}'$ still begin at different edges of $S$ and end at different edges of $T$, $\cl{C}'\in\cl{F}$.
	Moreover, since $W'_i,W'_j$ overlap at an edge, we have $\cl{C}'\not\in\cl{D}$.
	
	Thus $\cl{C}'\in\cl{S}$.

	Additionally, we claim that if we apply the above suffix swapping procedure (which we used to go from $\cl{C}$ to $\cl{C}'$) to the collection $\cl{C}'$, we recover $\cl{C}$.
	
	Indeed, for all $l\in [r]$ with $l < i$,
	the walk $W'_l = W_l$ does not intersect any other walk in $\cl{C}$
	at an edge, by the definition of $i$.
	Since the multiset of edges traversed by walks in $\cl{C}\setminus \set{W_l}$ and $\cl{C}\setminus \set{W'_l}$
	are the same,
	this means that $W'_l$ does not intersect any other walk in $\cl{C}'$ at an edge either.
	So
	$i$ is  the smallest index in $[r]$
	such that $W'_i$ intersects some other walk in $\cl{C}'$
	at an edge.
	Since the prefixes of $W'_i$ and $W_i$
	before edge $e$ are the same,
	we see that $e$ is also the first edge in $W'_i$
	which is contained in another walk of $\cl{C}'$.
	Then because $W_j'$ traverses edge $e$,
	and $W_l' = W_l$ for all $l\not\in \set{i,j}$,
	we get that $j > i$ is the smallest index such that $W_j'$
	overlaps with $W'_i$ at edge $e$. 
	Then, when we swap the suffixes of $W'_i$ and $W'_j$
	after the first appearance of $e$ on these walks,
	we recover $W_i$ and $W_j$ respectively,
	and so applying the suffix swapping procedure to
	$\cl{C}'$ produces the original collection $\cl{C}$ as claimed.

	So, the suffix swapping routine described above partitions $\cl{S}$ into distinct pairs.
	
	Suppose $\cl{C}$ and $\cl{C}'$ are paired up by the suffix swapping argument.
	Then $\cl{C}$ and $\cl{C}'$ traverse the same multiset of consecutive pairs of edges.
	Thus these collections
	\[\xi(\cl{C}) = \xi(\cl{C}')\]
	have the same weight.
	Since we work over a field of characteristic two, the above equation implies that each pair $(\cl{C},\cl{C}')$ of collections mapped to each other by suffix swapping satisfies
	\[\xi(\cl{C}) + \xi(\cl{C}') = 0.\]
	Since $\cl{S}$ is partitioned into such pairs, we have 
	\[\sum_{\cl{C}\in \cl{S}} \xi(\cl{C}) = 0.\]
	Together with the discussion from the beginning of the proof, this proves the claim.    
\end{proof}

\begin{remark}[Characteristic Two is Unnecessary]
	In the proof of \Cref{lem:lgv}, our argument used the fact that we work over a field of characteristic two.
	This restriction on the characteristic is only included for the sake of simplicity, and is not necessary to enumerate families of edge-disjoint walks.
If we instead worked over a field of odd characteristic, then all that changes is that terms in the expansion of the determinant from \Cref{lem:det-definition} come with a sign, and we can now pair up and cancel terms with opposite signs to prove a signed variant of \Cref{lem:lgv}.
\end{remark}

\begin{corollary}[Edge-Disjoint Walks]
	\label{cor:disjoint-summation}
	For any equal-size subsets of edges $S,T\sub E$, we have 
	\[\det~(I-X)^{-1}[S,T] = \sum_{\ell=1}^{\infty}\grp{\sum_{\cl{C}\in\cl{D}_\ell(S,T)}\xi(\cl{C})}.\]
\end{corollary}
\begin{proof}
	This follows by combining \Cref{lem:det-definition,lem:lgv}.
\end{proof}

\begin{lemma}[Edge-Disjoint Walks $\Rightarrow$ Edge-Disjoint Paths]
	\label{lem:walk-to-path}
	Let $S,T\sub E$ be subsets of edges of size $|S|=|T|=r$.
	If the graph $G$ contains $r$ edge-disjoint walks from $S$ to $T$, then $G$ also contains $r$ edge-disjoint paths from $S$ to $T$.
\end{lemma}
\begin{proof}
	Let $\langle W_1,\dots, W_r\rangle$ be a collection of edge-disjoint walks from $S$ to $T$ in $G$.
	For each index $i\in [r]$, let $e_i$ and $f_i$ be the first and last edges of $W_i$ respectively.
	Note that under these definitions, we have $S = \set{e_1, \dots, e_r}$ and $T = \set{f_1, \dots, f_r}$.
	
	For each $i\in [r]$, let $G_i$ be the subgraph of $G$ including only the edges traversed by $W_i$.
	Let $P_i$ be a shortest path from $e_i$ to $f_i$ in $G_i$.
	These paths are edge-disjoint, since they live in subgraphs on disjoint sets of edges.
	Thus  $\langle P_1, \dots, P_r\rangle$ is a collection of $r$ edge-disjoint paths from $S$ to $T$ in $G$, as desired.
\end{proof}

\begin{corollary}
	\label{lem:symbolic-det-char}
	Let $S, T\sub E$ be subsets of edges of size $|S| = |T| = r$.
	Then \[\det~(I-X)^{-1}[S,T]\]
	is a nonzero formal power series if and only if $G$ contains $r$ edge-disjoint paths from $S$ to $T$. 
\end{corollary}
\begin{proof}
	Suppose $\cl{C} = \langle P_1, \dots, P_r\rangle$ is a collection of $r$ edge-disjoint paths from $S$ to $T$ in $G$.
	Then the term
	$\xi(\cl{C})$ occurs in the expansion of 
	\begin{equation}
		\label{ix-inv-time}
		\det~(I-X)^{-1}[S,T]
	\end{equation}
	given by \Cref{cor:disjoint-summation}.
	Moreover, any collection of walks $\cl{C}'\neq \cl{C}$ from $S$ to $T$ has  weight $\xi(\cl{C}')\neq \xi(\cl{C})$, because $\cl{C}$ consists of edge-disjoint paths (so looking at the variables appearing in $\xi(\cl{C})$, we can recover $\cl{C}$ uniquely).
	Hence, no other term from the summation in \Cref{cor:disjoint-summation} produces the same monomial $\xi(\cl{C})$.
	So $\xi(\cl{C})$ appears in \Cref{ix-inv-time}
	with nonzero coefficient, which implies that the determinant from \Cref{ix-inv-time} is a nonzero formal power series.
	
	Suppose now that $G$ does not contain $r$ edge-disjoint paths from $S$ to $T$.
	The contrapositive of \Cref{lem:walk-to-path} implies that $G$ does not contain $r$ edge-disjoint walks from $S$ to $T$ either. 
	Then \Cref{cor:disjoint-summation} implies that \Cref{ix-inv-time}
	is the zero polynomial.
	This proves the claim.
\end{proof}

\begin{remark}[Directed Acyclic Graphs]
	\label{remark:DAGs}
	The arguments in this section  work with formal power series because when designing a generating function for edge-disjoint walks, we naturally run into infinite sums.
	From a pedagogical perspective, dealing with these infinite objects may make teaching these algorithms appear somewhat difficult (say, in an undergraduate course).
	One way to avoid this issue is to focus on the \textsf{APC} and \textsf{$k$-APC} problems in the special case of directed acyclic graphs (DAGs).
	
	DAGs are an interesting case for \textsf{APC} and \textsf{$k$-APC}, since we do not know of algorithms solving these problems on DAGs faster than the case of general directed graphs, and the best conditional lower bounds we have for these problems hold in the case of DAGs \cite{ap-bounded-mincut}.
	
	Focusing on the setting of DAGs leads to two simplifications.
	First, any walk in a DAG has length at most $n-1$. 
	Hence over DAGs, $X^\ell$ is the all-zeros matrix for $\ell\ge n$, so \Cref{prop:geo-series} shows that $\Gamma = (I-X)^{-1}$ is a  matrix whose entries are polynomials.
	Thus, determinants of submatrices of $\Gamma$ are just polynomials, instead of formal power series involving infinite sums.
	Second, any walk in a DAG is a path.
	Thus we can skip the step in \Cref{lem:walk-to-path} of going from edge-disjoint walks to edge-disjoint paths.
\end{remark}

\section{Connectivity}
\label{sec:APC}

\subsection{Random Evaluation}

Let $X$ be the symbolic edge-adjacency matrix defined in \Cref{sec:enumeration}.

For all pairs of edges $(e,f)$ in $G$, we introduce independent, uniform random values $a_{ef}$ over $\mathbb{F}$.
Let $A$ be the matrix obtained from $X$ by evaluating each variable $x_{ef}$ at $a_{ef}$.
That is, $A$ is the random $m\times m$ edge-adjacency matrix of $G$, defined by taking 
\[A = \begin{cases}
	a_{ef} & \text{if }\head(e)=\tail(f) \\
	0 & \text{otherwise.}
\end{cases}\]

\begin{lemma}
	\label{lem:random-det-char}
	Let $S, T\sub E$ be subsets of edges of size $|S| = |T| = r \le n-1$.
	Then \[\det~(I-A)^{-1}[S,T]\]
	is nonzero with probability at least $1-1/n^3$ if and only if 
	$G$ contains $r$ edge-disjoint paths from $S$ to $T$. 
\end{lemma}
\begin{proof}
	By \Cref{lem:symbolic-det-char}, the determinant
	\begin{equation}
		\label{eq:x-inv}
		\det~(I-X)^{-1}[S,T]
	\end{equation}
	is a nonzero formal power series if and only if $G$ contains $r$ edge-disjoint paths from $S$ to $T$.
	
	So suppose $G$ does not contain $r$ edge-disjoint paths from $S$ to $T$.
	Then \cref{eq:x-inv} is the zero polynomial, so its random evaluation \[\det~(I-A)^{-1}[S,T]\] 
	must vanish as claimed.
	
	Otherwise, suppose $G$ does contain $r$ edge-disjoint paths from $S$ to $T$.
	Then \cref{eq:x-inv} is a  nonzero formal power series.
	
	By the formula for the inverse of a matrix, we know that 
	\begin{equation}
		\label{eq:rational-inverse}
		(I-X)^{-1}[S,T] = \frac{(\adj(I-X))[S,T]}{\det~(I-X)}.
	\end{equation}
	
	Since $(I-X)$ has ones along the diagonal and its other entries have constant term zero, we know that $\det(I-X)$ is a polynomial with constant term $1$, so by \Cref{prop:inverse} the multiplicative inverse of $\det(I-X)$ is a well-defined formal power series.
	Thus, \Cref{eq:rational-inverse} can be viewed as an equality between two matrices of formal power series.
	
	For convenience, write  $Q = \det (I-X)$.
	Since $S$ and $T$ are sets of size $r$, by linearity of the determinant we have 
	\begin{equation}
		\label{eq:det-to-adj}
		\det~(I-X)^{-1}[S,T] = \frac{\det~(\adj(I-X))[S,T]}{Q^r}.
	\end{equation}
	By assumption, the left-hand side of \Cref{eq:det-to-adj} is nonzero.
	Consequently, the numerator
	\[\det~(\adj(I-X))[S,T]\]
	on the right-hand side of \Cref{eq:det-to-adj}
	must be a nonzero polynomial.
	Moreover, since each entry of $X$ has degree at most $1$, each entry of $\adj(I-X)$ has degree less than $m$, so this numerator polynomial has overall degree less than $r m$.
	Similarly, in the previous discussion we observed that $Q$ is a polynomial with constant term $1$, so
	the denominator 
	\[Q^r = \grp{\det~(I-X)}^r \]
	has constant term $1$ and is thus a nonzero polynomial as well.
	Since each entry of $X$ has degree at most $1$, this denominator polynomial has degree at most $rm$.
	
	The previous paragraph shows that the expression from \cref{eq:det-to-adj} is the ratio of two nonzero polynomials, each with degree at most $r m < nm$.
	
	Then by rational identity testing (\Cref{cor:rat-test}), the random evaluation 
	\[\det~(I-A)^{-1}[S,T]\]
	of 
	\Cref{eq:x-inv} over $\mathbb{F}$ is nonzero with probability at least $1-2nm/2^q$.
	
	Since we picked $q$ such that $2^q\ge 12n^6 > n^3\cdot (2nm)$, the desired result follows.
\end{proof}

\begin{lemma}[Connectivity via Rank]
	\label{lem:random-eval-rank}
	With high probability, for all $s,t\in V$ we have
	\[\lambda(s,t) = \rank~ (I-A)^{-1}[\eout(s),\ein(t)].\]
\end{lemma}
\begin{proof}
	Fix a pair of vertices $(s,t)$.
	Abbreviate $\lambda = \lambda(s,t)$.
	
	By \Cref{lem:random-det-char} and the definition of connectivity, with probability at least $1-1/n^3$, $\lambda$ is the largest nonnegative integer such that there exist subsets $S\sub \eout(s)$ and $T\sub \ein(t)$ of size $\lambda$ 
	with the property that 
	\begin{equation}
		\label{eq:inverse-expr}
		\det~(I-A)^{-1}[S,T]
	\end{equation}
	is nonzero.
	By definition, this is the rank of $(I-A)^{-1}[\eout(s),\ein(t)]$, so 
	\[\lambda = \rank~(I-A)^{-1}[\eout(s),\ein(t)]\]
	with probability at least $1-1/n^3$ for our fixed pair of vertices $(s,t)$.
	
	Then by the union bound, with probability at least $1-1/n$ we have 
	\[ \lambda(s,t) = \rank~ (I-A)^{-1}[\eout(s),\ein(t)]\]
	for all pairs of vertices $(s,t)$ in $G$, as desired.
\end{proof}

\begin{algorithm}[t]
	\caption{The algorithm solving \textsf{APC} from \cite{CheungLauLeung2013}.}
	\label{alg:APC}

	Compute the matrix $M = (I - A)^{-1}$.
	\label{step:invert}

	For each pair of vertices $(s,t)$, return
	\[\rank M[\eout(s),\ein(t)]\]
	as the value of $\lambda(s,t)$.
	\label{step:decode}
\end{algorithm}

\subsection{The Algorithm}

\apc*
\begin{proof}

	By \Cref{lem:random-eval-rank}, \Cref{alg:APC} solves \textsf{APC} correctly with high probability.
	It remains to bound the runtime of the algorithm.
	
	In step 1 of \Cref{alg:APC}, we compute $M$ by inverting an $m\times m$ matrix.
	By \Cref{matrix:inverse}, this matrix inversion can be performed in $\tilde{O}(m^\omega)$ time.
	
	In step 2 of \Cref{alg:APC}, we compute $\lambda(s,t)$ for each pair of vertices $(s,t)$ by computing the rank of a $\deg_{\text{out}}(s)\times \deg_{\text{in}}(t)$ matrix.
	By \Cref{matrix:rank}, this takes 
	\begin{equation}
		\label{eq:degree-bound}
		\sum_{s,t\in V} \deg_{\text{out}}(s)(\deg_{\text{in}}(t))^{\omega-1}
	\end{equation}
	time asymptotically.
	For each pair of vertices $(s,t)$, we have 
	\[\deg_{\text{out}}(s)(\deg_{\text{in}}(t))^{\omega-1} = (\deg_{\text{in}}(t))^{\omega-2}\cdot \deg_{\text{out}}(s)\deg_{\text{in}}(t) \le n^{\omega-2}\cdot \deg_{\text{out}}(s)\deg_{\text{in}}(t).\]
	By substituting this inequality  into \cref{eq:degree-bound}, and observing that the sum of in-degrees and sum of out-degrees are each equal to the number of edges $m$ in $G$, we have 
	\[\sum_{s,t\in V} \deg_{\text{out}}(s)(\deg_{\text{in}}(t))^{\omega-1} \le \sum_{s,t\in V}n^{\omega-2}\cdot \deg_{\text{out}}(s)\deg_{\text{in}}(t) = n^{\omega-2}m^2.\]
	Since $n-1\le m$, the runtime of this step is also upper bounded by $\tilde{O}(m^\omega)$.
	
	So we can solve \textsf{APC} in $\tilde{O}(m^\omega)$ time as claimed.
\end{proof}

\section{Bounded Connectivity}
\label{sec:k-APC}

The \textsf{APC} algorithm from \Cref{alg:APC} first
\begin{enumerate}
	\item inverts an $m\times m$ matrix, and then
	\item computes ranks of submatrices, whose dimensions depend on degrees of nodes in $G$.
\end{enumerate}
Even reading the matrix entries in these steps takes $\Omega(m^2)$ time.
To obtain a faster  algorithm for the \textsf{$k$-APC} problem, we modify these steps to work with much smaller matrices.

The first idea is to \emph{reduce degrees} in $G$ while preserving the values of small connectivities.
In \Cref{subsec:degree-reduction}, we present a simple transformation (from \cite[Section 5]{k-APC}) which decreases the degrees of nodes in $G$ to $k$, while preserving the $\min(k,\lambda(s,t))$ values.
Following this modification, we only need to compute ranks of $k\times k$ submatrices in step 2 above.

The second idea is to simplify $\Gamma = (I-X)^{-1}$ using a \emph{variable substitution} which still ensures that determinants of submatrices of $\Gamma$ can detect up to $k$ edge-disjoint paths in $G$.
We present this substitution in \Cref{subsec:substitution}.
This modification turns $\Gamma$ into a rank $kn$ matrix, which lets us replace the inversion of an $m\times m$ matrix in step 1 above with the easier inversion of a $kn\times kn$ matrix instead.

These two speed-ups combined then let us solve \textsf{$k$-APC} in $\tilde{O}((kn)^\omega)$ time.

\subsection{Degree Reduction}
\label{subsec:degree-reduction}

Let $G$ be the input graph on $n$ nodes and $m$ edges.
We modify $G$ to create a new graph.

For each vertex $v\in V$, we introduce two new nodes $v_{\t{in}}$ and $v_{\t{out}}$.
Then we replace each edge $(u,v)\in E$ with an edge $(u_{\t{out}}, v_{\t{in}})$.
For each $v\in V$, we also include $k$ parallel edges from $v$ to $v_{\t{out}}$, and $k$ parallel edges from $v_{\t{in}}$ to $v$.
Let $G_{\t{new}}$ be the new graph constructed in this way, and let $V_{\t{new}}$ and $E_{\t{new}}$ be its vertex and edge sets respectively.
We refer to the nodes in $V\subseteq V_{\t{new}}$ which were originally in $G$ as the \emph{original vertices}.
For $s,t\in V$, we still let $\lambda(s,t)$ denote the connectivity from $s$ to $t$ in the original graph $G$.
In the rest of this section, we let $\eout(s)$ and $\ein(t)$ denote the sets of edges exiting $s$ and entering $t$ in $G_\new$.

We write $n_\new = |V_\new| = 3n$ and $m_{\new} = |E_\new| = m+2kn$.

\begin{lemma}[Preserving Small Connectivities]
	\label{small:conn}
	For any $s,t\in V$, the connectivity from $s$ to $t$ in $G_\new$ is $\min(k,\lambda(s,t))$.
\end{lemma}
\begin{proof}
	Fix $s,t\in V$. 
	Given an $st$-path $P'$ in $G_\new$, we recover a unique $st$-path $P$ in $G$ by looking at the sequence of original vertices visited by $P'$.
	Using this construction, any collection of $r$ edge-disjoint paths from $s$ to $t$ in $G_\new$ recovers a collection of $r$ edge-disjoint paths from $s$ to $t$ in $G$.
	So the connectivity from $s$ to $t$ in $G_\new$ is at most $\lambda(s,t)$.
	
	Since $s$ has outdegree $k$ in $G_\new$, the connectivity from $s$ to $t$ in $G_\new$ is also at most $k$.
	
	Thus the connectivity from $s$ to $t$ in $G_\new$ is at most $\min(k,\lambda(s,t))$.
	Set $\lambda = \min(k,\lambda(s,t))$.
	
	By definition, there are edge-disjoint paths $P_1, \dots, P_\lambda$ in $G$ from $s$ to $t$.
	
	For each $i\in [\lambda]$, let $P'_i$ be the $st$-path in $G_{\t{new}}$ which passes through the same sequence of original vertices as $P_i$, and includes, for each edge $(u,v)$ in $P_i$, the $i^{\text{th}}$ parallel edge from $u_{\t{out}}$ to $v_{\t{in}}$ (we assume there is some fixed ordering among all such parallel edges).
	This is possible since $\lambda\le k$.
	Since the $P_i$ are edge-disjoint, the $P'_i$ are edge-disjoint as well.
	So the connectivity from $s$ to $t$ in $G_{\new}$ is at least $\lambda = \min(k,\lambda(s,t))$.

	Thus the connectivity from $s$ to $t$ in $G_{\new}$ is equal to $\min(k,\lambda(s,t))$, as claimed.
\end{proof}

\subsection{Low-Rank Edge-Adjacency}
\label{subsec:substitution}

Recall the definitions from \Cref{subsec:degree-reduction}.
We define the matrix $X$ from \Cref{sec:enumeration} with respect to the new graph $G_\new$ (so now rows and columns of $X$ are indexed by edges in $E_\new$).

For each pair $(e,j)\in E_\new\times [k]$ we introduce an indeterminate $y_{ej}$.

Similarly, for each pair $(j,f)\in [k]\times E_\new$ we introduce an indeterminate $z_{jf}$.

We define the $m_\new\times kn_\new$ matrix $Y$ by setting
\[Y[e,(v,j)] = \begin{cases}
	y_{ej} & \text{if }\head(e)=v \\
	0 & \text{otherwise.}
	
\end{cases}\]

Similarly, we define the $kn_\new\times m_\new$ matrix $Z$ by setting
\[Z[(v,j),f] = \begin{cases}
	z_{jf} & \text{if }\tail(f) = v\\
	0 & \text{otherwise.}
\end{cases}\]

These matrices are defined so that under the variable substitution
\[x_{ef} = \sum_{j=1}^k y_{ej}z_{jf}\]

the matrix $X$ simplifies to the \emph{low-rank} matrix $YZ$, as depicted in \Cref{fig:rank}.

Previously in \Cref{lem:symbolic-det-char}, we  characterized the existence of edge-disjoint paths in $G$, based off whether determinants of  submatrices of $(I-X)^{-1}$ were nonzero.
The following result shows that a similar characterization holds when we replace $X$ with $YZ$, provided we only care about routing up to $k$ edge-disjoint paths.

\begin{figure}
	\centering
	\begin{tikzpicture}[scale=3,
vtx/.style={
inner sep=0pt},
edge/.style={-stealth, thick,darkmidnightblue},
mvtx/.style={circle,draw=darkorange,fill=darkorange,inner sep=1pt, minimum width=1pt},
ledge/.style={thick, dotted, darkorange}]

    \def\top{0.3cm};
    \def\half{0cm};
    \def\bottom{-0.3cm};

    \def\stp{1cm};
    \def\sml{0.5cm};

    \def\extr{0.4cm};
    \def\extrw{0.6cm};

    \def\jump{2*\stp + 2*\extr + \extrw}

    \def\dnr{0.3cm}

    \node[vtx] (1) at (0,\half) {};
        \node[vtx] (2) at (\stp,\half) {};
        \node[vtx] (3) at (2*\stp,\half) {};

    \draw[edge] (1)--(2) node[midway, above] {$e$};
        \draw[edge] (2)--(3) node[midway, above] {$f$};

    \node (firstpath) at (\stp,\bottom-\dnr) {\textcolor{darkmidnightblue}{$X[e,f] = x_{ef}$}};


    \node (100) at (2*\stp + \extr,\half) {};
    \node (101) at (2*\stp + \extr + \extrw,\half) {};

    \draw[ultra thick,-latex,alizarin] (100) -- (101);


    \node[vtx] (1e) at (\jump, \half) {};
    \node[vtx] (2e) at (\jump + \stp, \half) {};

        \draw[edge] (1e) -- (2e) node[midway, above] {$e$} ;

    \node[mvtx] (e1) at (\jump + \stp + \sml, \top) {};
        \node[mvtx] (e2) at (\jump + \stp + \sml, \half) {};
        \node[mvtx] (e3) at (\jump + \stp + \sml, \bottom) {};

        \draw[ledge] (2e) -- (e1);
        \draw[ledge] (2e) -- (e2);
    \draw[ledge] (2e) -- (e3);

    \node[vtx] (3e) at (\jump + \stp + 2*\sml, \half) {};

        \draw[ledge] (e1) -- (3e);
        \draw[ledge] (e2) -- (3e);
    \draw[ledge] (e3) -- (3e);

    \node[vtx] (4e) at (\jump + 2*\stp + 2*\sml, \half) {};

    \draw[edge] (3e) -- (4e) node[midway, above] {$f$};

    \node (firstpath) at (\jump + \stp + \sml,\bottom-\dnr) {$\textcolor{darkmidnightblue}{(YZ)[e,f] 
    = y_{e\textcolor{darkorange}{1}}z_{\textcolor{darkorange}{1}f} + 
    y_{e\textcolor{darkorange}{2}}z_{\textcolor{darkorange}{2}f} +
    y_{e\textcolor{darkorange}{3}}z_{\textcolor{darkorange}{3}f}
    }$};

\end{tikzpicture}
	\caption{When we substitute $x_{ef} = y_{e1}z_{1f} + \dots + y_{ek}z_{kf}$ (pictured here for $k=3$) into $X$, we get the ``simpler'' matrix $YZ$.
		While powers of $X$ enumerate walks in $G$, powers of $YZ$ intuitively enumerate walks in a modified graph where after traversing an edge $e=(u,v)$, we have $k$ different versions of $v$ we can choose to go to.
		The $y_{ej}$ and $z_{jf}$ variables in this enumeration only keep track of the individual edges traversed and versions of vertices we pick, instead of recording all pairs of consecutive edges traversed like the $x_{ef}$ variables.
		This simpler enumeration suffices to solve \textsf{$k$-APC}.}
	\label{fig:rank}
\end{figure}

\begin{lemma}
	\label{lem:uptok-nonzero}
	Let $S,T\sub E_\new$ be subsets of edges with size $|S|=|T| = r\le k$. Then
	\[\det~(I-YZ)^{-1}[S,T]\]
	is a nonzero formal power series if and only if $G_{\new}$ has $r$ edge-disjoint paths from $S$ to $T$.
\end{lemma}
\begin{proof}
	Suppose $G$ does not contain $r$ edge-disjoint paths from $S$ to $T$.
	Then by \Cref{lem:symbolic-det-char}, the determinant 
		\[
		\det~(I-X)^{-1}[S,T]
		\]
	is identically zero as a power series.
	Consequently, the above expression remains zero even if we make the variable substitution
	\begin{equation}
		\label{eq:variable-substitution-low-rank}
		x_{ef} = \sum_{j=1}^k y_{ej}z_{jf}.
	\end{equation}
	Under this substitution, the matrix $X$ simplifies to $YZ$.
	Thus in this case
	\[\det~(I-YZ)^{-1}[S,T]\]
	is the zero polynomial, as claimed.

	Suppose now that $G$ does not contain  $r$ edge-disjoint paths from $S$ to $T$.
	
	By \Cref{cor:disjoint-summation}, we have
	\begin{equation}
		\label{eq:inv-expr}
		\det~(I-X)^{-1}[S,T] = \sum_{\ell=1}^\infty 
		\grp{\sum_{\cl{C}\in\cl{D}_{\ell}(S,T)}
			\xi(\cl{C})
		}.
	\end{equation}
	For each collection of walks $\cl{C}$, let $\tilde{\xi}(\cl{C})$ be the monomial resulting from substituting \Cref{eq:variable-substitution-low-rank}
	into the weight $\xi(\cl{C})$.
	Then we have 
	\begin{equation}
		\label{eq:detYZ}
		\det~(I-YZ)^{-1}[S,T] = \sum_{\ell=1}^\infty 
		\grp{\sum_{\cl{C}\in\cl{D}_{\ell}(S,T)}
			\tilde{\xi}(\cl{C})}.
	\end{equation}
	
	Let $\cl{P} = \langle P_1,\dots, P_r\rangle$ be a collection of edge-disjoint paths from $S$ to $T$ in $G$.
	
	For each $i\in[r]$, let $\cl{E}_i$ be the set of consecutive pairs of edges $(e,f)$ traversed by $P_i$.
	
	Then we have 
	\[\tilde{\xi}(\cl{P}) = \prod_{i=1}^r \prod_{(e,f)\in\cl{E}_i} \grp{\sum_{j=1}^k y_{ej}z_{jf}}.\]
	If we expand the product on the right-hand side of the above equation, we see that one of the monomials produced is of the form 
	\begin{equation}
		\label{eq:unique-monomial}
		\prod_{i=1}^r \prod_{(e,f)\in\cl{E}_i} y_{ei}z_{if}.
	\end{equation}
	Note that in this step, we are using the fact that $r\le k$.
	
	Because $\cl{P}$ is a collection of edge-disjoint paths, the variables appearing in the monomial from \cref{eq:unique-monomial}
	allow us to uniquely recover $\cl{P}$.
	In detail: for each index $i$, the edges $e$ for which the $y_{ei}$ variable appears recovers all edges in $P_i$, and because $P_i$ is a simple path we can recover the order of these edges as well.
	See \Cref{fig:disjoint} for an example of this unique recovery property in the case of $k=2$.
	
	Thus the monomial from \cref{eq:unique-monomial} appears with coefficient $1$.
	Hence \[\det~(I-YZ)^{-1}[S,T]\]
	is a nonzero formal power series as claimed.
\end{proof}

\begin{figure}
	\centering
	\begin{tikzpicture}[scale=2.5,
		vtx/.style={
			inner sep=0pt},
		edge/.style={-stealth, thick,darkmidnightblue},
		mvtx/.style={circle,draw=darkorange,fill=darkorange,inner sep=1pt, minimum width=1pt},
		ledge/.style={thick, dotted, darkorange}]
		
		\def\top{0.6cm};
		\def\half{0cm};
		\def\bottom{-0.6cm};
		
		\def\stp{1cm};
		\def\sml{0.5cm};
		
		\def\extr{0.8cm};
		\def\extrw{0.7cm};
		
		\def\smlvert{0.3cm};
		
		
		\node[vtx] (1) at (0,\top) {};
		\node[vtx] (2) at (\stp,\half) {};
		\node[vtx] (3) at (2*\stp,\top) {};
		\node[vtx] (4) at (0,\bottom) {};
		\node[vtx] (5) at (2*\stp,\bottom) {};

		\draw[edge] (1) -- (2) node[near end, above=4] {$e_1$};
		\draw[edge] (2) -- (3) node[near start, above=3.1] {$f_1$};;
		
		\draw[edge] (4) -- (2) node[near end, below=5] {$e_2$};
		\draw[edge] (2) -- (5) node[near start, below=2.7] {$f_2$};
		
		
		\node (100) at (2*\stp + \extr,\half) {};
		\node (101) at (2*\stp + \extr + \extrw, \half) {};
		
		\draw[ultra thick, -latex, alizarin] (100) -- (101);
		
		\def\jump{2*\stp + 2*\extr + \extrw};
		
		
		\node[vtx] (1s) at (\jump,\top) {};
		\node[vtx] (2s) at (\jump+\stp,\half) {};
		\node[vtx] (4s) at (\jump,\bottom) {};
		
		\node[mvtx] (v1) at (\jump+\stp + \sml,\smlvert) {};
		\node[mvtx] (v2) at (\jump+\stp + \sml,-\smlvert) {};

		\draw[edge] (1s) -- (2s) node[near end, above=4] {$e_1$};
		\draw[edge] (4s) -- (2s) node[near end, below=5] {$e_2$};
		
		\draw[ledge] (2s) -- (v1);
		\draw[ledge] (2s) -- (v2);
		
		\node[vtx] (2ss) at (\jump+\stp + 2*\sml,\half) {};
		
		\draw[ledge] (v1) -- (2ss);
		\draw[ledge] (v2) -- (2ss);
		
		\node[vtx] (3s) at (\jump+2*\stp + 2*\sml,\top) {};
		\node[vtx] (5s) at (\jump+2*\stp + 2*\sml,\bottom) {};
		
		\draw[edge] (2ss) -- (3s) node[near start, above=3.1] {$f_1$};
		\draw[edge] (2ss) -- (5s) node[near start, below=2.7] {$f_2$};
		
	\end{tikzpicture}
     \captionsetup{singlelinecheck=off}\caption{ Given edge-disjoint paths $P_1 = \langle e_1,f_1\rangle$ and $P_2=\langle e_2,f_2\rangle$ in $G$, the determinant of the matrix $(I-X)^{-1}[\set{e_1,e_2},\set{f_1,f_2}]$ enumerates this pair via the monomial $\xi(P_1,P_2) = x_{e_1f_1}\cdot x_{e_2f_2}$.
		The variables in this monomial provide enough information to uniquely recover $P_1$ and $P_2$. In contrast, for $k=2$, the determinant of $(I-YZ)^{-1}[\set{e_1,e_2},\set{f_1,f_2}]$ assigns this pair  weight 
		\[\tilde{\xi}(P_1,P_2) = (y_{e_11}z_{1f_1} + y_{e_12}z_{2f_1})(y_{e_21}z_{1f_2} + y_{e_22}z_{2f_2}).\]
		One of the terms in the expansion of the above product is $y_{e_11}z_{1f_1}\cdot y_{e_22}z_{2f_2}$. We can read this term as saying ``the first path $P_1$ traverses $e_1$ and $f_1$, and the second path $P_2$ traverses $e_2$ and $f_2$.'' So this monomial provides enough information to recover the pair of paths $\langle P_1,P_2\rangle$ as well.     }
	\label{fig:disjoint}
\end{figure}

We also use the following lemma that simplifies computation of $(I-YZ)^{-1}$. 

\begin{lemma}[Geometric Series Identity]
	\label{lem:reroute}
	We have
	\[(I - YZ)^{-1} = I + Y(I-ZY)^{-1}Z.\]
\end{lemma}
\begin{proof}
	Since every entry of $Y$ and $Z$ have zero constant term, the same is true for the matrices $YZ$ and $ZY$.
	Then by the geometric series formula of \Cref{prop:geo-series}, we have 
	\begin{equation}
		\label{eq:change-grouping}
		(I-YZ)^{-1} = I + (YZ) + (YZ)^2 + \dots = I + Y\grp{\sum_{\ell=0}^\infty (ZY)^\ell}Z.
	\end{equation}
	Applying \Cref{prop:geo-series} again, we have 
	\[\sum_{\ell=0}^\infty (ZY)^\ell = (I-ZY)^{-1}.\]
	Substituting the above equation into the rightmost side of \cref{eq:change-grouping} yields
	\[(I-YZ)^{-1} = I + Y(I-ZY)^{-1}Z\]
	as desired.
\end{proof}

\subsection{Random Evaluation}

We begin by defining random evaluations of the polynomial matrices $Y$ and $Z$.

For all pairs $(e,j)\in E_\new\times [k]$ and  $(j,f)\in [k]\times E_\new$ we introduce independent, uniform random values $b_{ej}$ and $d_{jf}$ respectively over $\mathbb{F}$.
Let $B$ and $C$ be the matrices obtained from matrices $Y$ and $Z$  under the evaluations $y_{ej} = b_{ej}$ and $z_{jf} = c_{jf}$ respectively.

That is, $B$ is the $m_\new\times kn_\new$ matrix defined by setting
\[B[e,(v,j)] = \begin{cases}
	b_{ej} & \text{if }\head(e)=v \\
	0 & \text{otherwise}
	
\end{cases}\]

and $C$ is the $kn_\new\times m_\new$ matrix defined by setting
\[C[(v,j),f] = \begin{cases}
	c_{jf} & \text{if }\tail(f) = v\\
	0 & \text{otherwise.}
\end{cases}\]

We also define subsets of edges
\[\eout = \bigcup_{s\in V} \eout(s) \quad\text{and} \quad\ein = \bigcup_{t\in V} \ein(t)\]
and let $\tilde{B} = B[\eout,\cdot]$ and $\tilde{C}=C[\cdot,\ein]$ be submatrices of $B$ and $C$ restricted to edges exiting and entering original vertices respectively.

\begin{lemma}
	\label{lem:low-rank-characterization}
	Let $S,T\sub E_\new$ be subsets of edges of size $|S|=|T| = r\le k$. Then
	\[\det~(I-BC)^{-1}[S,T]\]
	is nonzero with probability at least $1-1/n^3$ if and only if $G_{\new}$ contains  $r$ edge-disjoint paths from $S$ to $T$.
\end{lemma}
\begin{proof}
	By \Cref{lem:uptok-nonzero}, the formal power series
	\begin{equation}
		\label{yz-inv}
		\det~(I-YZ)^{-1}[S,T]
	\end{equation}
	is nonzero if and only if $G_{\new}$ contains $r$ edge-disjoint paths from $S$ to $T$.
	
	If $G_{\new}$ does not have $r$ edge-disjoint paths from $S$ to $T$, then \cref{yz-inv} is the zero polynomial, so its random evaluation 
	\[\det~(I-BC)^{-1}[S,T]\]
	vanishes as well.
	
	Suppose instead $G_{\new}$ does contain $r$ edge-disjoint paths from $S$ to $T$.
	We can write 
	\[(I-YZ)^{-1}[S,T] = \frac{(\adj(I-YZ))[S,T]}{\det~(I-YZ)}.\]
	The matrix $(I-YZ)$ has ones  along its diagonal, and all its other entries have zero constant term.
	Hence $\det(I-YZ)$ is a polynomial with constant term $1$, so by \Cref{prop:inverse} it has a multiplicative inverse over the ring of formal series.
	In particular, the above equation holds both for matrices of rational functions and formal power series.
	
	Write $Q = \det(I-YZ)$.
	By the above discussion, $Q$ is a nonzero polynomial.
	By linearity of the determinant, we have 
	\begin{equation}
		\label{nested:det}
		\det~(I-YZ)^{-1}[S,T]  = \frac{\det~ (\adj(I-YZ))[S,T]}{Q^r}.
	\end{equation}
	
	Since $Q$ is nonzero and $YZ$ is an $m_\new\times m_\new$ matrix where each entry has degree at most two, $Q^r$ has degree at most $2r\cdot m_\new \le 2k(m+2kn) < 6n^3$.
	
	Since $r$ edge-disjoint paths from $S$ to $T$ exist, \cref{eq:detYZ} is nonzero as a formal power series, so by \Cref{nested:det} the numerator $\det (\adj(I-YZ)[S,T])$ is a nonzero polynomial.
	Each entry of $I-YZ$ has degree at most two, so each entry of $\adj(I-YZ)$ has degree less than $2m_\new$, which implies that 
	$\det (\adj(I-YZ)[S,T])$ has degree at most 
	\[2m_\new r \le 2(m+2kn)n \le 6n^3.\]
	So by \Cref{cor:rat-test}, the expression
	\[\det~(I-BC)^{-1}[S,T]\]
	is nonzero in this case with probability at least $1-12n^3/(2^q)$.
	Since we picked $q$ large enough to satisfy $2^q\ge 12n^6 = n^3\cdot (12n^3)$, the desired result follows.
\end{proof}

\begin{lemma}[Small Connectivities via Rank]
	\label{lem:low-rank-correct}
	With high probability, for all $s,t\in V$, we have 
	\[\rank~(\tilde{B}(I-CB)^{-1}\tilde{C})[\eout(s),\ein(t)] = \min(k,\lambda(s,t)).\]
\end{lemma}
\begin{proof}
	Fix $s,t\in V$.
	Let $\lambda$ be the connectivity from $s$ to $t$ in $G_\new$.
	
	By \Cref{small:conn}, $\lambda\le k$.
	Thus by \Cref{lem:low-rank-characterization} and the definition of connectivity, with probability at least $1-1/n^3$, $\lambda$ is the largest nonnegative integer for which there exist subsets $S\sub\eout(s)$ and $T\sub\ein(t)$ of size $\lambda$ such that 
	\[\det~(I-BC)^{-1}[S,T]\]
	is nonzero.
	In other words, $\lambda$ is equal to the rank of $(I-BC)^{-1}[\eout(s),\ein(t)]$.
	
	By applying \Cref{lem:reroute} with the random evaluation sending $Y$ and $Z$ to $B$ and $C$ respectively, we have 
	\[(I-BC)^{-1} = I + B(I-CB)^{-1}C.\]
	Since $s,t\in V$ are original vertices, we know that $\eout(s)\cap\ein(t)=\emptyset$.
	Thus 
	\[(I-BC)^{-1}[\eout(s),\ein(t)] = (B(I-CB)^{-1}C)[\eout(s),\ein(t)].\]
	Of course we also have
	\[(B(I-CB)^{-1}C)[\eout(s),\ein(t)] = (\tilde{B}(I-CB)^{-1}\tilde{C})[\eout(s),\ein(t)].\]
	So with probability at least $1-1/n^3$, 
	\[\rank~(\tilde{B}(I-CB)^{-1}\tilde{C})[\eout(s),\ein(t)] = \min(k,\lambda(s,t))\]
	where we are using the fact from \Cref{small:conn} that $\lambda = \min(k,\lambda(s,t))$.
	The claim follows by a union bound over all $n^2$ pairs of original vertices $(s,t)$.
\end{proof}

\subsection{The Algorithm}

\begin{algorithm}[t]
	\caption{The algorithm solving \textsf{$k$-APC} from \cite{k-APC}.}
	\label{alg:k-APC}
	\setcounter{AlgoLine}{0}
	Compute the matrix $M = \tilde{B}(I-CB)^{-1}\tilde{C}$.
	\label{step:encode}
	
	For each pair $(s,t)$ of original vertices, return
	\[\rank M[\eout(s),\ein(t)]\]
	as the value for $\min(k,\lambda(s,t))$.
	\label{step:bounded-decode}
\end{algorithm}

\kapc*
\begin{proof}
	By \Cref{lem:low-rank-correct}, \Cref{alg:k-APC} correctly solves \textsf{$k$-APC} with high probability.
	It remains to bound the runtime of the algorithm.
	
	In step 1 of \Cref{alg:k-APC}, we compute $\tilde{B}(I-CB)^{-1}\tilde{C}$.
	The matrix $CB$ has rows and columns indexed by pairs $(v,i)\in V_\new\times [k]$.
	For any $(u,i)$ and $(v,j)$ in $V_\new\times [k]$, we have 
	\[CB[(u,i),(v,j)] = \sum_{e \in \eout(u)\cap\ein(v)} c_{ie}b_{ej}.\]
	If $v\in V$, then the summation in the right-hand side of the above equation is empty unless    $u=v_{\t{in}}$, in which case the sum has exactly $k$ terms (one for each parallel edge from $v_{\t{in}}$ to $v$).
	So computing the entries of $CB$ corresponding to this case takes $nk^2\cdot k = nk^3$ operations.

	Similarly, if $u\in V$, the sum in the above equation is empty unless
	$v=u_{\t{out}}$, in which case the sum consists of exactly $k$ terms.
	Computing the entries of $CB$ corresponding to this case takes $nk^3$ operations as well.
	
	For all other cases where $u,v\in V_\new\setminus V$, the sum consists of at most a single term.
	So we can compute the remaining entries of $CB$ with $((n_\new - n)k)^2 = 4n^2k^2$ operations.
	
	Since $k < n$ we have $nk^3 < (nk)^2$, and we can compute $CB$ in $\tilde{O}((kn)^2)$ time.
	Having computed $CB$, we can compute $(I-CB)^{-1}$ in $\tilde{O}((kn)^\omega)$ time by \Cref{matrix:inverse}.
	So step 1 of \Cref{alg:k-APC} takes $\tilde{O}((kn)^\omega)$ time overall.

	Step 2 of \Cref{alg:k-APC} involves computing ranks of $n^2$ separate $k\times k$ matrices,
	which by \Cref{matrix:rank}  takes $\tilde{O}(k^\omega n^2)$ time.
	
	So overall, the algorithm for \textsf{$k$-APC} takes $\tilde{O}((kn)^\omega)$ time as claimed. 
\end{proof}

\section{Conclusion}
\label{sec:conclusion}

In this paper, we presented alternate derivations of the $\tilde{O}(m^\omega)$ time algorithm for \textsf{APC} of \cite{CheungLauLeung2013}, and the $\tilde{O}((kn)^\omega)$ time algorithm for \textsf{$k$-APC} of \cite{k-APC}.
Our approach works by testing, via random evaluation, whether certain generating functions enumerating edge-disjoint families of paths are nonzero or not.
We conclude this paper by pointing out some connections between this perspective and other classical arguments in mathematics and computer science, and highlighting the current main open problems concerning the complexity of computing connectivities.

\paragraph*{Combinatorics}

Our proof for \Cref{lem:lgv} obtains a generating function for edge-disjoint walks by pairing up monomials corresponding to intersecting walks, and arguing their contributions cancel.
This reasoning
is essentially identical to the proof of the classic Lindstr\"{o}m-Gessel-Viennot  lemma from combinatorics, which is often used in mathematics to enumerate families of disjoint lattice paths.
We refer the reader to \cite[Chapter 29]{fromthebook} for an accessible exposition of this theorem and some of its applications. 
For additional examples of this technique of pairing up and cancelling extraneous terms in matrix algebra, see \cite{Zeilberger1985}.

\paragraph*{Algebraic Algorithms}

Our exposition of the \textsf{APC} and \textsf{$k$-APC} algorithms works by interpreting certain determinants combinatorially, as generating functions for families of edge-disjoint walks.
Previous work has also leveraged combinatorial interpretations of the determinant (as polynomials instead of formal power series however) to construct interesting arithmetic circuits \cite{comb-det,det-hom}. 
In particular, the sign-reversing involution designed in \cite[Proof of Theorem 1]{comb-det} is very similar to the suffix swapping argument we use in the proof of \Cref{lem:lgv}.

The general approach of solving graph problems by testing whether certain enumerating polynomials are nonzero is now a common technique in graph algorithms: we refer the reader to \cite{det-sieving} and the citations therein for previous examples of this paradigm in the literature.

\paragraph*{Open Problems}

\begin{description}
	\item[Faster Algorithms]
	The main open question in this area is: can we solve \textsf{APC} faster? 
	In particular, does there exist some constant $\varepsilon > 0$ such that \textsf{APC} can be solved in general directed graphs in $O(n^{4-\varepsilon})$ time?
	Even obtaining such an algorithm in the special case of directed acyclic graphs would be a major breakthrough.
	
	\hfill
	
	For constant $k$, it is conjectured that \textsf{$k$-APC} requires $n^{\omega-o(1)}$ time to solve, because algorithms for \textsf{$k$-APC} can be used to solve a problem known as \textsf{Boolean Matrix Multiplication}, which researchers currently do not know how to solve faster than integer matrix multiplication. 
	Under this conjecture, the $\tilde{O}((kn)^\omega)$ algorithm for \textsf{$k$-APC} is near-optimal for constant $k$. 
	However, this does not rule out the possibility of algorithms with better dependence on $k$.    
	
	\hfill
	
	Can we obtain faster algorithms for \textsf{$k$-APC} when $k = n^{\delta}$ for some small constant $\delta > 0$?
	
	\hfill
	
	\item[Better Lower Bounds]
	In the \textsf{4-Clique} problem, we are given an undirected graph $G$ on $n$ nodes, and are tasked with determining if $G$ contains four vertices which are mutually adjacent.
	The best conditional lower bounds for \textsf{APC} come from observing that \textsf{4-Clique} reduces to \textsf{APC} \cite[Section 4]{ap-bounded-mincut}.
	Using the current best algorithms for rectangular matrix multiplication \cite[Table 1]{mmult-life}, the fastest 
	known algorithm for \textsf{4-Clique} takes $O(n^{3.251})$ time \cite[Theorem 1]{EisenbrandGrandoni2004}.
	Even if we conjecture this runtime time is optimal, the resulting lower bound for \textsf{APC} is a far cry from the $n^{4+o(1)}$ runtime we have for \textsf{APC} in dense graphs.
	Moreover, if $\omega = 2$ this lower bound for \textsf{APC} weakens to $n^{3-o(1)}$.
	
	\hfill
	
	Can we show better lower bounds for \textsf{APC}, which are supercubic even if $\omega = 2$?
	
	\hfill
	
	\item[Faster Verification]
	Instead of solving \textsf{APC} directly, it would also be interesting to obtain better deterministic and randomized \emph{verifiers} for \textsf{APC}.
	A verifier for \textsf{APC} is given the input to the problem, and additionally receives claims for the values of $\lambda(s,t)$ for all pairs of vertices $(s,t)$, as well as some small proof string which can be thought of as ``evidence'' that the claimed values are correct.
	The verifier reads all of these inputs, and then must determine if the claimed values are correct or not.
	The goal is to get a verifier which is correct (with high probability) and runs as quickly as possible.
	
	\hfill
	
	The fastest known deterministic verifier for \textsf{APC} runs in $O(n^{3.251} + n^{2.5}\sqrt{m})$ time \cite{Trabelsi23}.
	It would be interesting to improve this runtime to $O(n^{3.251})$ to match the fastest known algorithm (and deterministic verifier runtime) for \textsf{4-Clique}.
	
	\hfill
	
	It would also be interesting to obtain faster \emph{randomized} verifiers for \textsf{APC}.
	Currently, no randomized verifiers running faster than the deterministic verifier discussed above are known for \textsf{APC}.
	This is surprising, since the best lower bound for \textsf{APC} comes from the \textsf{4-Clique} problem, and the \textsf{4-Clique} problem admits a randomized verifier running in near-optimal $\tilde{O}(n^2)$ time \cite[Section 4]{MA-protocol}.
	
	\hfill
	
	Does \textsf{APC} admit a faster randomized verifier? 
	Alternatively, can we explain the lack of fast randomized verifiers for \textsf{APC} by obtaining an efficient reduction to \textsf{APC} from some problem which we do not believe admits fast randomized verifiers?
	
	\hfill
	
	\item[Vertex-Connectivity Variants]
	Given vertices $s$ and $t$ in graph $G$, the \emph{vertex connectivity} from $s$ to $t$, denoted by $\nu(s,t)$, is the maximum number of internally vertex-disjoint paths from $s$ to $t$ in $G$.
	One can study the \textsf{All-Pairs Vertex Connectivity (APVC)} and \textsf{$k$-Bounded All-Pairs Vertex Connectivity ($k$-APVC)} problems as variants of \textsf{APC} and \textsf{$k$-APVC} where we are tasked with computing $\nu(s,t)$ and $\min(k,\nu(s,t))$ respectively, for all pairs $(s,t)$.

	\hfill
	
	In directed graphs, \textsf{APVC} and \textsf{$k$-APVC} reduce to \textsf{APC} and \textsf{$k$-APVC} respectively, so it might be easier to design faster algorithms for the former problems instead of the latter problems. 
	For example, it is known that \textsf{$k$-APVC} can be solved in $\tilde{O}(k^2n^\omega)$ time \cite[Theorem 5]{k-APC}, which is faster than the $\tilde{O}((kn)^\omega)$ runtime for \textsf{$k$-APC} if $\omega > 2$.
	
	\hfill
	
	In undirected graphs, although \textsf{APC} can be solved in $\tilde{O}(n^2)$ time, \textsf{APVC} is not known to admit a near-quadratic time algorithm.
	In fact, the best known lower bounds for \textsf{APC} in directed graphs also hold for \textsf{APVC} in undirected graphs \cite{vertex-conn-lower-bound}.
	
	\hfill
	
	So again, it might be easier to find faster algorithms for \textsf{APVC} in undirected graphs instead of tackling the general problem of \textsf{APC} in directed graphs.
	For example, \textsf{APVC} in undirected graphs can be solved in $\tilde{O}(m^2)$ time \cite[Section 3]{Trabelsi23}, which is faster than the $\tilde{O}(m^\omega)$ runtime for \textsf{APC} if $\omega > 2$, and admits a deterministic $O(n^{3.251})$ time verifier \cite[Lemma 2.2]{Trabelsi23}, which is faster than the known $O(n^{3.521} + n^{2.5}\sqrt{m})$ time deterministic verifier for \textsf{APC}.
	
	\hfill
	
	\item[Derandomization]
	
	The algorithms we presented for \textsf{APC} and \textsf{$k$-APC} in this work are randomized. 
	Can we solve these problems with respective $\tilde{O}(m^\omega)$ and 
	$\tilde{O}((kn)^\omega)$ runtimes \emph{deterministically}? 
	
\end{description}

\printbibliography

\end{document}